\documentclass[orivec]{llncs}

\input{psfig.sty}

\usepackage{times}
\usepackage{url}
\usepackage{subfigure}
\usepackage{amssymb,amsmath}
\usepackage{graphicx}
\usepackage{color}
\usepackage{wrapfig}
\usepackage{xspace}
\usepackage{MnSymbol}
\usepackage{cite}

\newcommand{\blink}[1]{\textnormal{\texttt{#1}}}
\newcommand{\NP}[0]{\blink{NP}\xspace}
\newcommand{\Oh}{{\ensuremath{\mathcal{O}}}}
\newcommand{\eps}{\ensuremath{\varepsilon}\xspace}

\DeclareMathOperator{\true}{true}
\DeclareMathOperator{\false}{false}
\newcommand{\df}{\textit}

\let\doendproof\endproof
\renewcommand\endproof{~\hfill\qed\doendproof}

\begin{document}

\title{On Embeddability of Buses in Point Sets}

\authorrunning{Bruckdorfer et al.}

\author{
  Till Bruckdorfer\inst{1}
  \and
  Michael Kaufmann\inst{1}
  \and
  Stephen Kobourov\inst{2}
  \and
  Sergey Pupyrev\inst{2,3}
}

\institute{
  Wilhelm-Schickard-Institut f\"ur Informatik, Universit\"at T\"ubingen, Germany
  \and
  Department for Computer Science, University of Arizona, USA
  \and
  Institute of Mathematics and Computer Science, Ural Federal University, Russia
}

\maketitle

\begin{abstract}
  Set membership of points in the plane can be visualized by connecting corresponding points
  via graphical features, like paths, trees, polygons, ellipses.
  In this paper we study the \emph{bus embeddability problem} (BEP): given a set of colored points
  we ask whether there exists a planar realization with one horizontal straight-line segment per color,
  called bus, such that all points with the same color are connected with vertical line segments to their bus.
  We present an ILP and an FPT algorithm for the general problem.
  For restricted versions of this problem, such as when the relative order of buses is predefined,
  or when a bus must be placed above all its points, we provide efficient algorithms.
  We show that another restricted version of the problem can be solved using 2-stack pushall sorting.
  On the negative side we prove the \NP-completeness of a special case of BEP.
\end{abstract}

\section{Introduction}
\label{sec:introduction}

Visualization of sets is an important topic in graph drawing and information
visualization and the traditional approach relies on representing overlapping sets
via Venn diagrams and Euler diagrams~\cite{simonetto2009fully}.
When more than a handful sets are present, however, such diagrams become difficult to interpret and alternative approaches,
such as compact rectangular Euler diagrams are needed~\cite{riche2010untangling}.

Often the geometric position of the elements of the sets are
prescribed as points in the plane.
The task is to emphasize the sets where the elements belong to.
In visualization approaches for set memberships of items on maps,
this is done by connecting points from the same set by corresponding
lines (LineSets~\cite{journals/tvcg/AlperRRC11}),
tree structures (KelpFusion~\cite{meulemans2013kelpfusion}),
and enclosing polygons (BubbleSet~\cite{journals/tvcg/CollinsPC09} or
MapSets~\cite{efrat2014mapsets}).

\begin{figure}[t]
\subfigure[\label{fig:pointset-a}]{\includegraphics[width=0.24\textwidth]{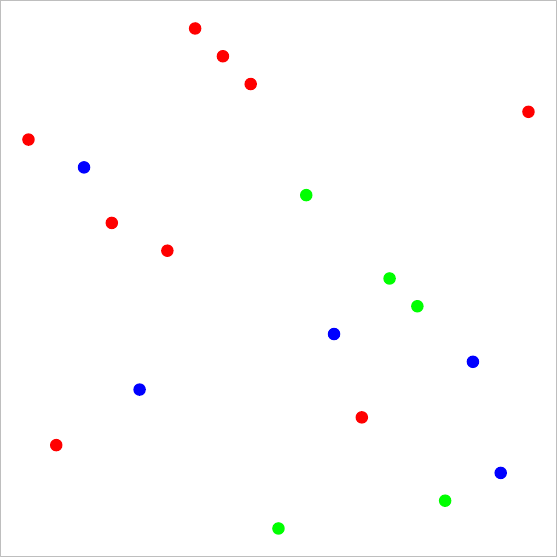}}
\hfill
\subfigure[\label{fig:pointset-b}]{\includegraphics[width=0.24\textwidth]{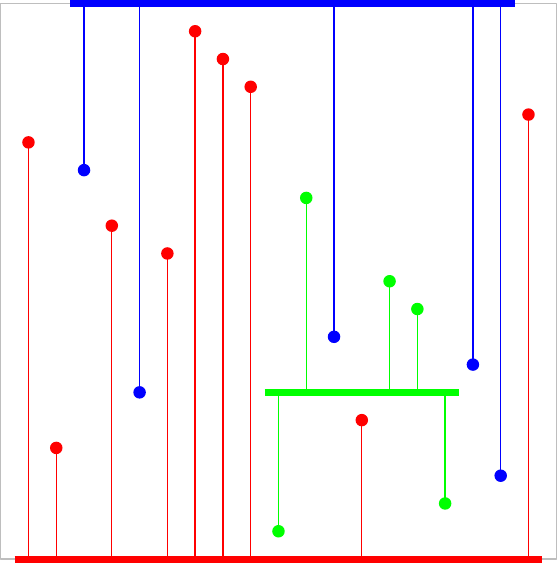}}
\hfill
\subfigure[\label{fig:pointset-c}]{\includegraphics[width=0.24\textwidth]{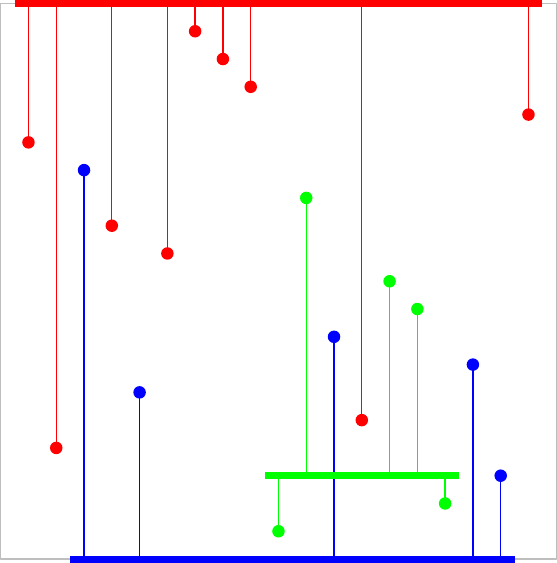}}
\hfill
\subfigure[\label{fig:pointset-d}]{\includegraphics[width=0.24\textwidth]{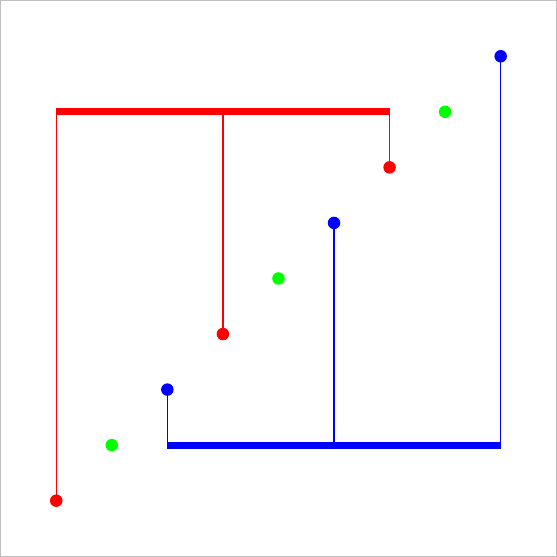}}
\caption{\subref{fig:pointset-a} Fixed positions of points,
where points with the same color belong to the same set.
\subref{fig:pointset-b}~A planar bus realization for this setting,
while \subref{fig:pointset-c}~is a non-planar bus realization.
\subref{fig:pointset-d}~A point set without any planar bus realization.}
	\label{fig:pointset}
\end{figure}

We consider a unified version of the tree-structure approach
using a model that has been applied before for drawing orthogonal buses
known from VLSI design~\cite{Thompson80acomplexity,DBLP:books/el/leeuwen90/Lengauer90}.
Our goal is a membership visualization of points in sets by a tree-structure that consists of
a single horizontal segment, called \emph{bus}, to which all the points
from the same set are connected by vertical segments, called \emph{connections};
see Fig.~\ref{fig:pointset} for planar and non-planar versions.
We assume the sets to be given by single-colored points, such that in the final
visualization, called \emph{bus realization}, every point of the same color
is connected to exactly one bus associated with this color.
The objective is to find a position for each bus,
such that crossings of buses with connections are avoided, called \emph{planar} bus realization.
We call this the \emph{bus embeddability problem} (BEP).
Such a simple visualization scheme makes it very easy to recognize the sets and label them,
by placing a label inside each bus (if the bus is drawn thick enough), or
directly above/next to the bus.

\paragraph{Related Work.}
Buses have been used, in a more general form, for visualizing degree-restricted hypergraphs.
Ada et al.~\cite{DBLP:journals/corr/abs-cs-0609127} used horizontal and vertical buses in bus realizations,
where the points (representing hypervertices contained in at most four hyperedges) were not predefined in the plane.
They asked whether a given hypergraph admits a non-planar bus realizations (allowing connections to cross each other) and showed that the problem is NP-complete.
In contrast, if a planar embedding is given, a planar bus realization can be constructed on
a $\Oh(n) \times \Oh(n)$ grid in $\Oh(n^{3/2})$
time~\cite{DBLP:conf/ciac/BruckdorferFK13}.
These types of problems also have connections to rectangular drawings, rectangular duals and visibility graphs,
since the edges of the incidence graph of a hypergraph enforce visibility constraints in the
bus realizations~\cite{He93onfinding,DBLP:journals/dcg/TamassiaT86}.

Another related approach is visualization based on graph supports of hypergraphs.
Here the goal is to connect the vertices in such a way that each hyperedge
induces a connected
subgraph~\cite{DBLP:journals/jgaa/BuchinKMSV11,DBLP:journals/jda/BrandesCPS12,DBLP:conf/swat/KlemzMN14}.
Supported hypergraph visualizations inspired edge-bundling and confluent
layouts as alternative visualizations for
cliques~\cite{Gansner06improvedcircular,DBLP:journals/corr/cs-CG-0212046,PNBH11}.

A solution to the BEP problem can be viewed as planar tree support for hypergraphs,
and this problem is related to Steiner trees~\cite{Hwang92thesteiner},
where the goal is to connect a set of points in the plane while minimizing the sum
of edge lengths in the resulting tree; this is a classic NP-complete problem~\cite{thecomplexitygarey}.
Hurtado et al.~\cite{DBLP:conf/gd/HurtadoKKLASS13} considered planar supports for
hypergraphs with two hyperedges such that the induced subgraph for every hyperedge
and the intersection is a Steiner tree. Their objective was to minimize the sum of edge lengths,
while allowing degree one or two for the hypervertices.
BEP is even more closely related to rectilinear Steiner trees~\cite{Ganley1999161},
where the Euclidean distance is replaced by the rectilinear distance;
constructing rectilinear Steiner trees is also NP-complete~\cite{gareyrectilinear1977}.
A single trunk Steiner tree~\cite{Chen:2002:RST:505348.505366}
 is a path which contains all vertices of degree greater than one.
This is a variant that is solvable in linear time.
BEP for a single set is the single trunk rectilinear Steiner
tree problem, where we ignore the minimization of the sum of the edge lengths.
Thus BEP can be seen as a simultaneous single-trunk rectilinear Steiner tree problem.
The fact that a bus placement influences the placement of other buses makes the problem hard.

Consider the input to BEP along with a box that encloses all the points.
If in BEP the buses extend to the right boundary of this box, or both to the left and
right boundary of this box, then this problem corresponds to backbone boundary
labeling and can be efficiently solved~\cite{DBLP:conf/gd/BekosCFHKNRS13}.
In backbone boundary labeling, the problem is to orthogonally connect points by a horizontal
backbone segment leading to a label placed at the boundary.
In this setting it is always possible to split the problem into two independent
subproblems, which is impossible in our case.

BEP is also related to the classical \emph{point set embeddability problem},
where given a set of points along with a planar graph, we need to determine
whether there exists a mapping of vertices to points such that
the resulting straight-line drawing is planar. The general decision problem is \NP-hard~\cite{Cabello06}.
In the variant of orthogeodesic point set embedding, Katz et al. proved that
deciding whether a planar graph can be embedded using only orthogonal edge
routing is \NP-hard~\cite{DBLP:conf/gd/KatzKRW09}.

\paragraph{Our Results.}
In Section~\ref{sec:Preliminaries} we solve BEP when the relative order of the buses is prescribed; we also show that BEP is fixed-parameter tractable (FPT) with respect to the number of colors.
In Section~\ref{sec:An ILP for BEP} we formulate an integer linear programming (ILP) formulation
for BEP and show some experimental results.
In Section~\ref{sec:Restricted Busproblem} we restrict BEP (when a bus must be above all its points, or a bus must be either at its topmost or bottommost point) and describe efficient algorithms for these settings. Another restricted version of the problem is shown to be equivalent to the problem of sorting a permutation, which is called 2-stack pushall sorting.
Finally we prove that BEP is \NP-complete, even for just two points per color, if points may not lie on buses.

\section{Preliminaries}
\label{sec:Preliminaries}

We begin with some definitions.
Suppose we are given a set of points ${\cal P}=\{p_1,\dots,p_n\}$
and colors ${\cal C}=\{c_1,\dots,c_k\}$
together with a function $f:{\cal P} \longrightarrow {\cal C}, f(p)=c$.
For simplicity, we assume that no two points share a coordinate in the input point set,
although in some illustrations the input points might violate this assumption.
The bus embeddability problem (BEP) asks,
whether there is a planar bus realization with one horizontal bus per color.
BEP is a decision problem, but in our descriptions whenever the answer is affirmative we also
compute a drawing. We refer to such a drawing as a \df{solution of BEP}. In the negative
case, we say that BEP has no solution.

A point $p$ has x-coordinate $x(p)$, y-coordinate $y(p)$, and color $f(p)$.
In a bus realization we have connections only between a point $p$ and a bus $c$ of the same color,
that is, $c=f(p)$. We denote by $f^{-1}(c)$ the set of points with color $c$.
Bus $c$ naturally extends from the x-coordinate $x_l(c) = \min \{x(p) | p \in f^{-1}(c)\}$
of the leftmost point to the x-coordinate $x_r(c) = \max \{x(p) | p \in f^{-1}(c)\}$
of the rightmost point of $f^{-1}(c)$. We call $[x_l(c),x_r(c)]$ the \emph{span} of $c$,
which is predefined by the input points. The y-coordinate of a bus $c$ is denoted by $y(c)$,
which is the only parameter to be determined for a solution for BEP.

Note that BEP is trivial when there are at most two colors: it is always possible
to place one bus at the top and the other (if exists) at the bottom
of the drawing. Thus in the following we assume $k > 2$.
For more than two colors, the relative order of the buses is important; see Fig.~\ref{fig:pointset}.
Suppose the y-order of the buses is prescribed. The next lemma shows that
one can check an existence of a solution for BEP respecting the order.

\begin{lemma}\label{lem:correctness}
There is a $\Oh(n \log n)$-time algorithm that,
given an order of buses, tests whether there exists a solution for BEP respecting the order.
\end{lemma}

\begin{proof}
Suppose we are given an order $c_1 < \dots < c_k$ of the buses from bottom to top.
We use discrete values for the y-coordinates increasing from bottom to top, where
a unit is $1/n$ of the y-distance of two consecutive points.
We first present a simpler $\Oh(n^2)$-time algorithm, and
then describe how to speed it up.

Recall that the span of every bus is defined by an input point set; hence, we only show how to
choose y-coordinates of the buses.
The first bus, $c_1$, is placed at y-coordinate $y(c_1)=0$, and
all the points of color $c_1$ are connected to the bus. Assume that
bus $c_{i-1}$ is placed at y-coordinate $y(c_{i-1})$ and is connected to all its points.
We place $c_i$ at $y(c_i) = y(c_{i-1}) + 1$ unit and check if the bus crosses a
previously drawn (vertical)
segment. If it does cross a segment, then we shift $c_i$ one unit upwards by increasing $y(c_i)$ and
repeat the procedure. Once the bus is placed without crossings, we connect it to the
corresponding points. Consider the vertical segment of a point $p$ of color $c_i$.
It is easy to see that if $y(p) \ge y(c_i)$, then the segment cannot cross a previously
placed bus $c_j$ for $j<i$.
If $y(p) < y(c_i)$ and the vertical segment crosses a bus, then such a crossing
is unavoidable in any solution respecting the given order. Hence, we may stop the algorithm reporting
that no solution exists. Otherwise, we proceed with the next color.

The above algorithm can easily be implemented in quadratic time. However,
we can do better using the following observation: Every bus is placed at its
bottommost ``valid'' y-coordinate, that is, the one that does not produce crossings with
previously placed buses.
To find such a y-coordinate efficiently for each color,
we store all points of the already processed colors in a data structure $D$ that supports
the range operation such as ``extracting minimum/maximum on a given range''. For every color $c_i$, we
extract a point with the maximum y-coordinate in the range corresponding to the span of $c_i$.
The bus of $c_i$ is placed at the maximum of the extracted y-coordinate and the y-coordinate
of bus $y(c_{i-1})$. Then all the points of color $c_i$ are added to $D$.
A balanced tree (e.g., a segment tree) providing logarithmic complexity for insert and extract operations
is sufficient for our needs.
\end{proof}

In general the correct order of the buses for a planar bus realization is not known.
One can apply Lemma~\ref{lem:correctness} for each of the $k!$
possible bus orders, which yields an $\widetilde{\Oh}(k!)$-time\footnote{$\widetilde{O}$ hides polynomial factors.}
algorithm for BEP. Next, we improve the running time with an algorithm
providing deeper insight into the structure of the problem.

\begin{lemma}\label{lem:2n}
There is a $\widetilde{\Oh}(2^k)$-time algorithm for BEP.
\end{lemma}

\begin{proof}
We solve a given instance of BEP using dynamic programming. Let us call a \df{state}
a pair $(h, B)$, where $0 \le h \le n+1$ is an integer and $B$ is a
subset of ${\cal C} = \{c_1, \dots, c_k\}$. By a solution for a state $(h, B)$ we mean
a (planar) bus realization consisting of buses for every color $c\in B$ such that
the topmost bus has y-coordinate $h$. If such a solution exists, we
write $F(h, B) = \true$, and otherwise $F(h, B) = \false$.
It is easy to see that a solution for the original BEP problem exists if and only if
$F(h, {\cal C})=\true$ for some $0\le h \le n+1$.

We reduce the problem to solving it for ``smaller'' states, that are the
states with fewer elements in $B$. As a base case, we set $F(h, B) = \true$
for all $0\le h \le n+1$ and $|B|=1$. To compute a value for a
state $F(h, B)$ with $|B|>1$, we consider a color $c^*\in B$. Let
$h^* = \max\{y(p)| f(p)\in B\setminus\{c^*\} \text{ and } x_l(c^*) \le x(p) \le x_r(c^*) \}$,
that is, the largest (topmost) y-coordinate of a point of color $B\setminus\{c^*\}$
laying in the span of $c^*$. It follows from the proof of Lemma~\ref{lem:correctness} that
the bus for $c^*$ should be placed at y-coordinate $h^*$.
Thus, $F(h, B)$ is set to $\true$ if (a)~$h \ge h^*$ and (b)~there exists a solution
for a state $(h', B\setminus\{c^*\})$ for some $h' < h$. We stress here that
in order to compute $F(h, B)$, one needs to consider every color of $B$ as a potential $c^*$.
There are $n2^k$ different states, and a computation for a single state clearly
takes a polynomial number of steps.
\end{proof}

The above result shows that the BEP problem is fixed-parameter tractable with respect to $k$, that is,
it can be efficiently solved for a small number of buses.
Note that in Section~\ref{sec:Combination of BEP restrictions}
we prove that BEP is \NP-complete; hence, it is unlikely that a polynomial-time (in terms of $k$)
algorithm exists.

\section{An ILP for BEP}
\label{sec:An ILP for BEP}

In this section we present an integer linear programming (ILP) formulation for BEP
that produces a planar bus realization if one exists. The ILP
also minimizes the amount of ink in a solution, that is, the sum of all segment lengths.

\begin{lemma}
A solution for BEP can be computed by an ILP.
\end{lemma}

\begin{proof}
In a preprocessing step we compute the span of every bus $c \in {\cal C}$.
As mentioned earlier, it remains to compute the y-coordinate variable $y(c)$ of every bus $c$.
To this end, we introduce a planarity constraint for every point $p \in {\cal P}$ within the
span of bus $c$ having a different color.
The pairs $(p, c), c \neq f(p)$ are called \df{conflicting}. Conflicting pairs $(p,c)$ are
stored in a matrix ${\cal J}$ and induce the constraint
$(y(p)<y(c)$ and $y(f(p)) < y(c))$ or $(y(p) > y(c)$ and $y(f(p)) > y(c))$.
The matrix ${\cal J}$ can be computed in $\Oh(kn)$ time, where $n$=$|{\cal P}|$ and $k$=$|{\cal C}|$.
In order to minimize the amount of ink, we sum up the lengths of all connections
and ignore the lengths of buses, as those are determined by the input.
\begin{eqnarray}
 \min && \sum_{c \in {\cal C}} \sum_{f(p)=c} |y(c) - y(p)|\notag\\
 s.t. && (y(p) < y(c) \vee y(f(p)) > y(c)) \wedge (y(p) > y(c) \vee y(f(p)) < y(c)) \quad \forall (p,c) \in {\cal J}\notag\\
      && 0 \leq y(c) \leq \max_{p \in {\cal P}} \{y(p)\} + 1\notag
\end{eqnarray}
Since absolute value (resp. ``or'') needs one more variable and $3$ constraints for every
point (resp. for every conflicting pair)
\footnote{$\min \sum |a-b| \Leftrightarrow \min \sum e, e \geq a-b, e \geq b-a, e \geq 0$;
$(a<b) \vee (c<d) \Leftrightarrow a-b<eM, c-d<(1-e)M, e \in \{0,1\}, M=\infty$},
the final ILP has $n$+$k$+$2|{\cal J}|$ variables and $3n$+$k$+$6|{\cal J}|$ constraints.
\end{proof}

\begin{figure}[t]
 \centering
 	\includegraphics[width=0.99\textwidth]{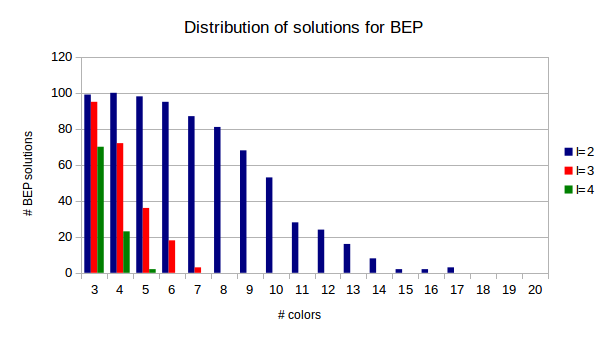}
 	\caption{The percentage of solutions for BEP for a random point set of size $n=kl$
with $l=2,3,4$ points per color out of $k=3,\dots,20$ colors.}
	\label{fig:result2}
\end{figure}

In order to get a feeling about the probability that a point set admits a solution of BEP,
we ran an experiment with the ILP, implemented with the Gurobi solver~\cite{gurobi}.
We considered point sets with $k=3,\dots,20$ colors and with $l=2,3,4$ points per color.
We randomly placed the points on a $1024 \times 768$ area.
For each pair $(l,k)$ we counted the number of BEP solutions out of $100$ instances; see Fig.~\ref{fig:result2}.
The remaining instances were infeasible.
For a fixed number of points, $l$, the number of solutions for BEP decreases with increasing the number of colors, $k$.
It decreases faster the higher $l$ is.
On the other hand for a fixed number of colors, $k$, the number of solutions for BEP also decreases with
increasing number of points, $l$. Hence, studying two points per color promises to be sufficiently interesting.
Thus, as the base case for further analysis, we initially consider two points per color, before dealing
with the general case, where in real instances solutions rarely exist.
It is possible that much more solutions exist if we allow only few crossings,
but all non-planar settings are left as open problems.

\section{Efficiently Solvable BEP Variants}
\label{sec:Restricted Busproblem}

In this section we consider three variants of BEP, which can be solved in polynomial time.
A bus $c$ is called \df{top} (resp., \df{bottom})
if all of its points are below (resp., above) the bus,
that is, $y(c) \geq y(p)$ (resp., $y(c) \leq y(p)$) for all $p \in f^{-1}(c)$.
We distinguish between buses that are above (below) of their points and buses that
pass through one of their points.
A top-bus is a \df{$\sqcap$-bus} if $y(c) > y(p)$ for all $p \in f^{-1}(c)$ (Fig.~\ref{fig:type-a}),
while it is a \df{$\lefthalfcap$-bus} if $y(c) = y(p)$ for a point $p$ with
$y(p) = \max \{y(q) | q \in f^{-1}(c)\}$ (Fig.~\ref{fig:type-c}).
Similarly we define a \df{$\sqcup$-bus} and a \df{$\lefthalfcup$-bus}; see Figs.~\ref{fig:type-b}~and~\ref{fig:type-d}.
A bus, whose type is none of the four types from above, is called a \df{center-bus}.
The variant of BEP where only buses of the types in $S \subseteq \{\sqcap, \sqcup, \lefthalfcap, \lefthalfcup\}$
are allowed to use is denoted by $S$-BEP.

\begin{figure}[b!]
\subfigure[\label{fig:type-a}]{\includegraphics[width=0.15\textwidth]{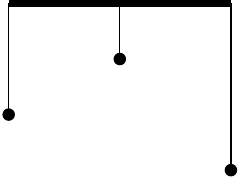}}
\hfill
\subfigure[\label{fig:type-b}]{\includegraphics[width=0.15\textwidth]{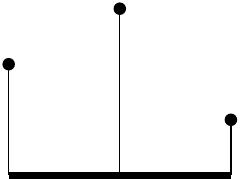}}
\hfill
\subfigure[\label{fig:type-c}]{\includegraphics[width=0.15\textwidth]{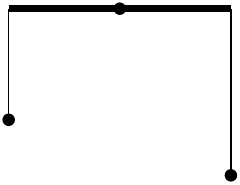}}
\hfill
\subfigure[\label{fig:type-d}]{\includegraphics[width=0.15\textwidth]{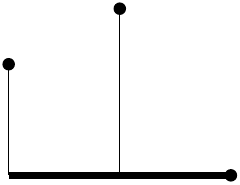}}
\caption{Illustration of \subref{fig:type-a} $\sqcap$-bus,
\subref{fig:type-b} $\sqcup$-bus, \subref{fig:type-c} $\lefthalfcap$-bus, and
\subref{fig:type-d} $\lefthalfcup$-bus.}
	\label{fig:type}
\end{figure}

In Section~\ref{sect:1bep} we study $\sqcap$-buses and provide an algorithm for $\sqcap$-BEP.
The same algorithm obviously solves the $\sqcup$-BEP variant. Next we consider
$\lefthalfcap$-buses and $\lefthalfcup$-buses. Note that $\lefthalfcap$-BEP and $\lefthalfcup$-BEP are trivial, since
every $\lefthalfcap$-bus (resp., $\lefthalfcup$-bus) is uniquely defined by its span and the topmost (bottommost)
point. Hence, we investigate and design an efficient algorithm for the $(\lefthalfcap, \lefthalfcup)$-BEP variant.
Finally in Section~\ref{sec:diagonal}, we examine the general BEP for a specific point set, where
all points lie on a diagonal. We show that the variant of the problem is equivalent to
a longstanding open problem (resolved very recently) of sorting a permutation with a series of two stacks.

\subsection{$\sqcap$-BEP}
\label{sect:1bep}
Here, we present an algorithm that decides in polynomial time whether
a drawing with $\sqcap$-buses exists for a given input,
and constructs such a drawing if one exists.

\begin{theorem}\label{lm:sqcap}
There exists an $\Oh(n \log n)$-time algorithm for $\sqcap$-BEP.
\end{theorem}

\begin{proof}
For ease of presentation, we first assume that the input consists of two points per color, that is,
$k=n/2$, and provide a simple quadratic-time implementation. Later we generalize the algorithm and improve
the running time. Intuitively, the algorithm sweeps a line from bottom to top and processes the
points in increasing order of y-coordinates.
At every step, we keep all the vertical segments of the ``active'' colors (the ones without a bus)
in the correct left-to-right order. If two vertical segments of the same color are adjacent
in the order, then we can draw the corresponding bus and remove the color and its vertical segments.
Otherwise, all the active vertical segments have to be ``grown'' until we reach the next point. It is easy to see
that a solution exists if and only if the set of active colors is empty after processing all the points.

More formally, the points are processed one-by-one in increasing order of their
y-coordinates. The points are stored in an array sorted by x-coordinate, that is, we have
$(p_1, \dots, p_n)$ with $x(p_1) < \dots < x(p_n)$. At each iteration,
a new point is inserted into the array in the position determined by its x-coordinate.
Then the array is modified (or simplified) so that the pairs of points of
the same color that are adjacent in the array
are removed. That is, if $f(p_i) = f(p_{i+1})$ for some $1 \le i < n$, then we get a new array
$(p_1, \dots, p_{i-1}, p_{i+2}, \dots, p_n)$.
The simplification is performed as long as the array contains monochromatic adjacent points.
After this step the algorithm proceeds with the next point.
For every color $c$, we keep the value $y^*(c)$, which is equal to the y-coordinate $y(p), p\in f^{-1}(c')$
of the point of color $c'$, whose insertion into the array induced the removal of points $f^{-1}(c)$ from the array.
If the algorithm ends up with a non-empty array, then we report that no solution exists.
Otherwise, the y-coordinate of the resulting bus of color $c$ is $y^*(c)+\eps$,
where $\eps > 0$ is sufficiently small to avoid overlaps between the buses.
An example of the algorithm is illustrated in Fig.~\ref{fig:2bendtop}.

\begin{figure}[tb]
\centering
    \includegraphics[]{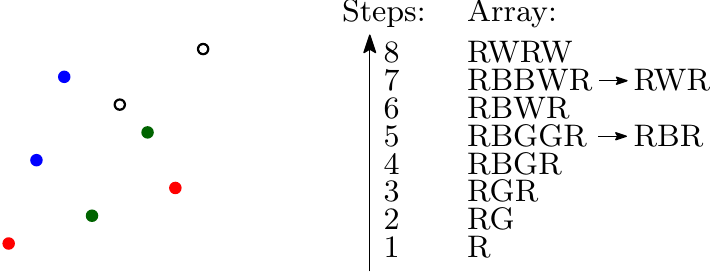}
\caption{Running the algorithm from Lemma~\ref{lm:sqcap} on a given point set with red (R), green (G),
blue (B), and white (W) pairs of points. Since the resulting array is not empty, there is no solution
for the instance. Notice that removing any of the colors yields an instance with a solution.}
\label{fig:2bendtop}
\end{figure}

\noindent
\emph{Correctness.} The correctness follows from the observation that the algorithm chooses
the lowest ``available'' y-coordinate for every bus, that is, the one that does not induce
a crossing between the bus and vertical segments of other colors. Indeed, if at any step of the
algorithm we get a color pattern $R, \dots, B, \dots, R$ in the array formed by red (R)
and blue (B) points and
the second blue point $p$ has not been processed yet, then clearly in any solution the red vertical
segments reach the y-coordinate of $p$. Hence, it is safe to ``grow'' the segments. On the other hand,
if processed points form a color pattern $RR$ (that is, two consecutive points of the same color), then
there is a solution connecting the corresponding vertical segments at the current y-coordinate. The two
points can be removed from consideration, as they cannot create crossings with the subsequent buses.
It is also easy to see that the algorithm minimizes ink of the resulting drawing.

\noindent
\emph{Running time.} At every iteration of the algorithm, we need to insert a new point into
the sorted array and then run the simplification procedure. Point insertion takes
$\Oh(n)$ time and removal of a pair of points from the array can also be done in $\Oh(n)$ time.
Since every pair is removed only once, the total running time is $\Oh(n^2)$.

To get down to $\Oh(n \log n)$ time, we use a balanced binary tree instead of an array to
store the points. The tree is sorted by the x-coordinates of the points; hence, insertion/removal
of a point takes $\Oh(\log n)$ time. Note that after inserting/removing a point, the only
potential candidate pairs for simplification are the point's neighbors that can be found in
$\Oh(\log n)$ time.
Again, every point is inserted/removed only once; thus, the total running time is $\Oh(n \log n)$.

Finally, we observe that the algorithm can be generalized to handle multiple points per color.
To this end, we change the simplification step so that the points are removed only if
they form a contiguous subsequence in the array (tree), containing all points of this color.
Hence we need to know the number of points for each color, which can be done with a linear-time scan of the input.
It is easy to see that the proof of correctness can be appropriately modified and the running time remains the same.
\end{proof}

\subsection{($\lefthalfcap, \lefthalfcup$)-BEP}
\label{sect:2bep}

We present an algorithm
that decides in polynomial time whether ($\lefthalfcap, \lefthalfcup$)-BEP has a solution
for a given input, and constructs a drawing if one exists.

\begin{theorem}\label{lm:gammal}
There exists an $\Oh(n^2)$-time algorithm for ($\lefthalfcap, \lefthalfcup$)-BEP.
\end{theorem}

\begin{proof}
The span of every bus is predefined by the input, while the y-coordinate
has precisely two options. We show that ($\lefthalfcap, \lefthalfcup$)-BEP can be
modeled by 2-SAT, and thus is efficiently solvable.
For ease of presentation, we first assume that the input consists of two points per color
and describe a simple quadratic-time algorithm.

The algorithm creates a variable $x_c$ for every color $c \in {\cal C}$.
The value of $x_c$ is \emph{true} if $c$ is a $\lefthalfcap$-bus, and it is \emph{false} if $c$
is a $\lefthalfcup$-bus. Then for every pair of colors $c, c'$, the algorithm creates
a clause for the 2-SAT instance when the corresponding buses induce a crossing.
Building the clauses with respect to the relative position of points
is a straight-forward procedure; 3 examples are illustrated in Fig.~\ref{fig:buildclauses}.

\begin{figure}[tb]
\centering
    \includegraphics[width=0.80\textwidth]{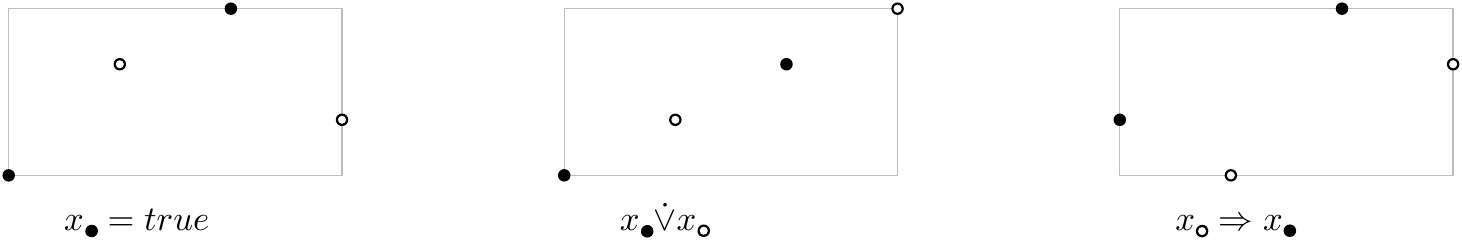}
\caption{Three examples for creating clauses for two colors black and white.}
\label{fig:buildclauses}
\end{figure}

Specifically we create clauses according to the following analysis.
Let $R(c)$ be the smallest enclosing rectangle of the points $p_c,q_c$ of color $c$.
By symmetry, we may assume that $p_c$ appears in the left bottom corner, while $q_c$ appears in
the right top corner of $R(c)$.

We distinguish the cases when
\begin{itemize}
\item[(1)] points $p_{c'},q_{c'}$ are in the top left, bottom right corner of $R(c')$ or whether
\item[(2)] points $p_{c'},q_{c'}$ are in the bottom left, top right corner of $R(c')$.
\end{itemize}

In each of the two cases we consider the $8$ subcases, which are
\begin{itemize}
\item[(a)] $R(c')$ intersects only the top boundary of $R(c)$,
\item[(b)] $R(c')$ intersects only the bottom boundary of $R(c)$,
\item[(c)] $R(c')$ intersects only the right boundary of $R(c)$,
\item[(d)] $R(c')$ intersects only the left boundary of $R(c)$,
\item[(e)] $R(c')$ contains the top right corner of $R(c)$,
\item[(f)] $R(c')$ contains the bottom right corner of $R(c)$,
\item[(g)] $R(c')$ contains the top left corner of $R(c)$,
\item[(h)] $R(c')$ contains the bottom left corner of $R(c)$.
\end{itemize}

\begin{table}[h]
\centering
\setlength{\arraycolsep}{1.5pt}
$
\begin{array}{c|c|c|c|c|c|c|c|c}
\text{cases} & a & b & c & d & e & f & g & h \\
\hline
1 & x_c=f & x_c=t & x_c=t & x_c=f & x_{c'}=t & x_c=t & x_c=f & x_{c'}=f \\
\hline
2 & x_c=f & x_c=t & x_c=t & x_c=f & x_c \dot{\vee} x_{c'} &
x_{c'} \Rightarrow x_c & \overline{x_{c'}} \Rightarrow \overline{x_c} & x_c \dot{\vee} x_{c'}\\
\end{array}
$\vspace{0.2cm}
\caption{For each of the cases (a)-(h) from above we build a clause depending
on the configuration (1)-(2) from above, where $t$ stands for true and $f$ for false.}
\label{tab:clausetable}
\end{table}

\noindent
\emph{Correctness.} The correctness follows from the complete case analysis by the rules of
Table~\ref{tab:clausetable}.

\noindent
\emph{Running time.} We remark that for the $n^2/4$ pairs of colors, we create $\Oh(n^2)$ clauses,
each clause in constant time by a case analysis.
This results in a 2-SAT instance with $k$ variables $x_c, c \in {\cal C}$ and $\Oh(n^2)$ clauses.
We solve this instance in linear time~\cite{journals/ipl/AspvallPT79}
and the solution determines the drawing: $c$ is drawn as a $\lefthalfcap$-bus, if the value of $x_c$ is \emph{true},
otherwise $c$ is drawn as a $\lefthalfcup$-bus.

We can generalize this idea to the case of more points per color.
In the general case the y-coordinate of a bus again has precisely two options.
In contrast to the case with two points per color we check several points (not only the leftmost
or rightmost point) of color $c'$ for their position with respect to the points of color $c$,
since points lie not necessarily in corners of the enclosing rectangle.
\end{proof}

\subsection{Diagonal BEP}
\label{sec:diagonal}

Here we consider a \df{diagonal} point set in which all
points lie on a single diagonal line and there are two points per color.
We assume that the point set is \df{separable}, that is, there is a straight line
separating every pair of points having the same color; see Fig.~\ref{fig:2-stack-ex2}.
This specific arrangement can be naturally described in terms of permutations.
Assuming that the colors are numbered from $1$ to $k$ in the order along the
diagonal from bottom to top, the input is described by a permutation $\pi=[\pi(1), \dots, \pi(k)]$ on
$\{1, \dots, k\}$. Such an instance is called \df{diagonal} $\pi$-BEP.

It turns out that this variant of BEP is closely related to the well-studied topic
of sorting a permutation with stacks introduced by Knuth in the 1960's~\cite{Knuth:1997:ACP:260999}.
We next show that diagonal $\pi$-BEP has a solution if and only if $\pi$ can
be sorted with 2 stacks in series. The problem of deciding whether a permutation
is sortable with 2 stacks in series is a longstanding open problem and it has been
conjectured to be \NP-complete several times~\cite{journals/combinatorics/Bona02}.
Only very recently a polynomial-time algorithm has been developed~\cite{DBLP:conf/stacs/PierrotR14,DBLP:journals/corr/abs-1303-4376}. It is an indication
that even the restricted variant of BEP is highly non-trivial. Next we prove the equivalence.

First observe that for a diagonal point set with 2 points per color, a top-bus (bottom-bus)
can be transformed to a center-bus. For every color $c$, there are no points of different color
within the span of $c$ above the topmost point of $c$. Hence, we may only consider center-buses
in the variant of BEP.

For the 2-stack sorting problem,
given a permutation $\pi$, we want to sort the numbers to the identity permutation $[1,\dots,k]$ with two
stacks $S_I, S_{II}$ using the following operations:
\begin{itemize}
 \item $\alpha_i:$ read the next element $i$ from input $\pi$ and push it on the first stack $S_I$;
 \item $\beta_i:$ pop the topmost element $i$ from $S_I$ and push it on $S_{II}$;
 \item $\gamma_i:$ pop the topmost element $i$ from $S_{II}$ and print it to the output.
\end{itemize}

\begin{figure}[tb]
 \centering
 	\includegraphics[width=1.00\textwidth]{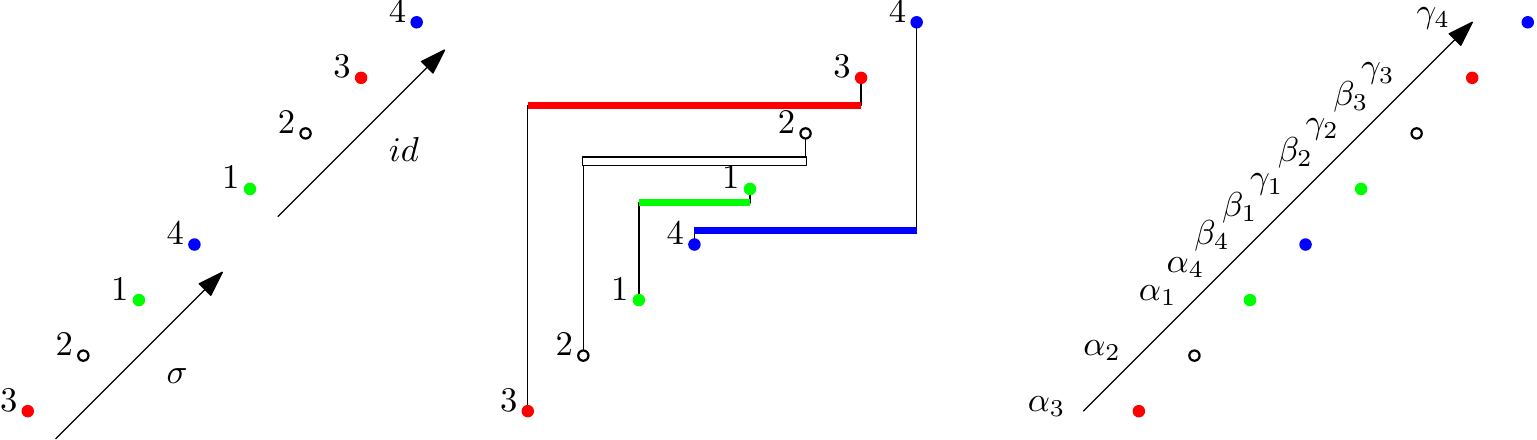}
 	\caption{A diagonal point set with a solution for BEP and the regarding sorting sequence.}
	\label{fig:2-stack-ex2}
\end{figure}

To proof of the equivalence between 2-stack sorting and bus embeddability, we note that
the first operation, $\alpha_i$, corresponds to the left vertical segment of color $i$,
the second one, $\beta_i$, is the bus of $i$, while $\gamma_i$ corresponds to the right
vertical segment of the color; see Fig.~\ref{fig:2-stack-ex2} and Fig.~\ref{fig:2-stack}.
A crossing in the drawing correspond to an ``invalid'' sorting operation in which
either a non-topmost element is moved from $S_I$ to $S_{II}$ (a crossing to the ``left'' of the diagonal), or
a non-topmost element is moved from $S_{II}$ to the output (a crossing to the ``right'' of the diagonal).
Hence, sorting sequences of the operations for $\pi$ are in one-to-one
correspondence with planar bus realization for
the point set. Since the point set is separable, all the elements of $\pi$ will be pushed to
$S_I$ before any of the elements is popped to the output. This is called 2-stack \df{pushall}
sorting~\cite{DBLP:journals/corr/abs-1303-4376} and is considered next in more detail.

We can describe a sequence of the operations by a word $w \in \{\alpha, \beta, \gamma\}^{3n}$,
where every operation appears $n$ times.

For example
$w = \alpha_3 \alpha_2 \alpha_1 \alpha_4 \beta_4 \beta_1 \gamma_1 \beta_2 \gamma_2 \beta_3 \gamma_3 \gamma_4$
is a sorting word for $\pi_1 = 3214$ to $\pi_2=1234$ with two stacks, see Table~\ref{tab:sort}.

\begin{table}[h]
\centering
\setlength{\arraycolsep}{2.5pt}
$
\begin{array}{c|c|c|c|c}
\text{operation} & \text{input} & S_I & S_{II} & \text{output}\\
\hline
 & 3214 &  &  & \\
\hline
\alpha_3 & 214 & 3 &  & \\
\hline
\alpha_2 & 14 & 23 &  & \\
\hline
\alpha_1 & 4 & 123 &  & \\
\hline
\alpha_4 &  & 4123 &  & \\
\hline
\beta_4 &  & 123 & 4 & \\
\hline
\beta_1 &  & 23 & 14 & \\
\hline
\gamma_1 &  & 23 & 4 & 1\\
\hline
\beta_2 &  & 3 & 24 & 1\\
\hline
\gamma_2 &  & 3 & 4 & 12\\
\hline
\beta_3 &  &  & 34 & 12\\
\hline
\gamma_3 &  &  & 4 & 123\\
\hline
\gamma_4 &  &  &  & 1234\\
\hline
\end{array}
$\vspace{0.2cm}
\caption{Permutation $[3, 2, 1, 4]$ is sortable with two stacks.}
\label{tab:sort}
\end{table}
\vspace{0.5cm}

A word $w$ also encodes the input and output of a sequence by subscripts, when disregarding the subscripts of
the beta operation. For example $s(w)=32141234$ is the sequence of subscripts for $w$.

We may restrict this sequence of operations
to only $\alpha$ and $\gamma$ operations, denoted by $w|\{\alpha,\gamma\}$.
We say $w$ is a \emph{pushall word}, if $s(w|\{\alpha,\gamma\})=\pi_1\pi_2$.
The word $w' = \alpha_3 \alpha_2 \alpha_1 \beta_1 \gamma_1 \alpha_4 \beta_4 \beta_2 \gamma_2 \beta_3 \gamma_3 \gamma_4$
also sorts $\pi_1$ to $\pi_2$ with two stacks, but $w'$ is not a pushall word,
since $s(w') = 32114234 \neq \pi_1\pi_2$.

Now we assume we are given $2n$ points on a diagonal respecting the order $\pi_1,\pi_2$.
We denote by $\pi_1 (\pi_2)$ the order of the first (second) appearance of the elements.
Every output word $w$ of the 2-stack-pushall-sortable algorithm describes the
sorting from $\pi_1$ to $\pi_2$.

The 2-stack-sorting algorithm takes as input $\pi_1$ and $\pi_2$ and returns in $\Oh(n^2)$ time
one sorting word of $E = \{w : w \text{ sorts } \pi_1 \text{ to } \pi_2\}$.
If such a word $w$ exists, then we can construct a planar bus realization with center buses of the embedded points $\pi_1\pi_2$ according to $w$ as follows. We apply one of the following 3 rules on the letters
of $w$. We process $w$ letter by letter and read along $3n$ imaginary slots on the diagonal.

\begin{itemize}
 \item[$\alpha_i$] the next slot of the diagonal is point $i$ with a connection going up.
 \item[$\beta_i$] the next slot of the diagonal is taken by the horizontal segment from the end of the
connection of point $i$, then crossing the diagonal.
 \item[$\gamma_i$] the next slot of the diagonal is point $i$ with a connection down to its
horizontal segment, extended such that this connection meets perpendicular.
\end{itemize}

This drawing is planar:
\begin{itemize}
 \item any crossing of two edges incident to $i,j$ to the left of the diagonal
comes from the sequence $\dots, \alpha(i),\dots, \alpha(j), \dots, \beta(i),\dots$,
which means push $i$ on $S_I$, then push $j$ on $S_I$ and then pop $i$ from $S_I$,
which is impossible since $i$ is not the topmost element of $S_I$.
 \item any crossing of two edges incident to $i,j$ to the right of the diagonal
comes from the sequence $\dots, \beta(i),\dots, \beta(j), \dots, \gamma(i),\dots$,
which means push $i$ on $S_{II}$, then push $j$ on $S_{II}$ and then pop $i$ from $S_{II}$
(and print $i$ to the output), which is impossible since $i$ is not the topmost element of $S_{II}$.
\end{itemize}

\begin{figure}[tb]
 \centering
 	\includegraphics[width=0.70\textwidth]{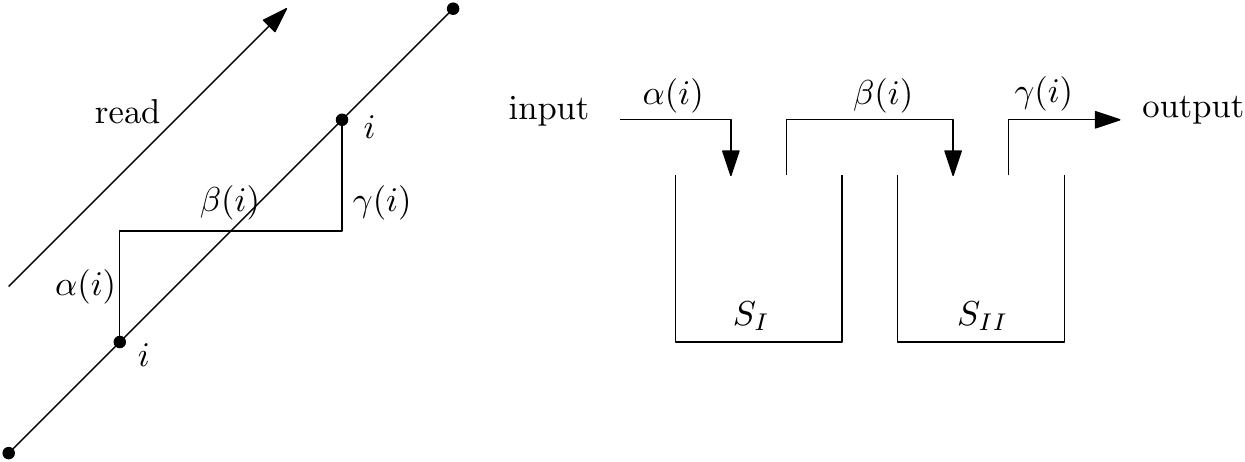}
 	\caption{A correspondence between 2-stack sorting and a planar bus realization.}
	\label{fig:2-stack}
\end{figure}

The construction from a planar bus realization with center buses of a diagonal point $\pi_1\pi_2$ set to a sorting word $w$ for $\pi_1\pi_2$
is just traversing the diagonal from bottom to top and simultaneously building incrementally the
sorting word $w$. We start with $w=\lambda$, where $\lambda$ is the empty word. If the next item on the diagonal is the first
appearance of a letter $i$, we set $w = w \circ \alpha_i$. If the next item on the diagonal is a crossing
of the edge connecting the two points of $i$, we set $w = w \circ \beta_i$. If the next item on the diagonal
is the second appearance of a letter $i$, we set $w = w \circ \gamma_i$. It is easy to see that this
word $w$ sorts $\pi_1$ to $\pi_2$. This finishes the proof.

\begin{theorem}
Diagonal $\pi$-BEP has a solution if and only if $\pi$ is 2-stack pushall
sortable. This can be checked in $\Oh(n^2)$ time.
\end{theorem}

\section{Hardness of BEP}
\label{sec:Combination of BEP restrictions}

In this section we consider BEP$^{\eps}$, where $\eps > 0$ is the additional input number
indicating the minimum allowed distance between points and their bus. We prove that
BEP$^{\eps}$ is \NP-complete even for 2 points per color.

We can easily verify a possible solution using Lemma~\ref{lem:correctness};
thus BEP$^{\eps}$ is in the class \NP.
We then show that ($\sqcap$,$\sqcup$)-BEP$^{\eps}$ for 2 points per color is \NP-hard.
To prove the hardness of ($\sqcap$,$\sqcup$)-BEP$^{\eps}$, we reduce from
planar 3-SAT~\cite{DBLP:journals/siamcomp/Lichtenstein82}, which is 3-SAT, where
an instance is represented by a graph whose vertices represent variables and clauses
and whose edges represent containment of variables in clauses.
The most important module of the construction is a \emph{chain link},
which is also a gadget for replacing variables.
It consists of two points on a common horizontal line that will be connected by a bus.
We replace the edges of the graph by chains consisting of nested chain links
and replace the clause vertices by a big construction of points,
that allows two specific points to be connected via a bus using only one of
three choices, cf. Fig.~\ref{fig:gadgetsPlanar3SAT}.
We use the input $\eps$ to be able to block some choices for this bus.
We first restrict ourselves to ($\sqcap, \sqcup$)-BEP$^{\eps}$
and drop the ``no points share a coordinate" restriction.
We finally transform the construction into the ``no points share a coordinate"
setting and allow also center-buses.

\begin{figure}[tb]
 \centering
 \includegraphics[width=1.00\textwidth]{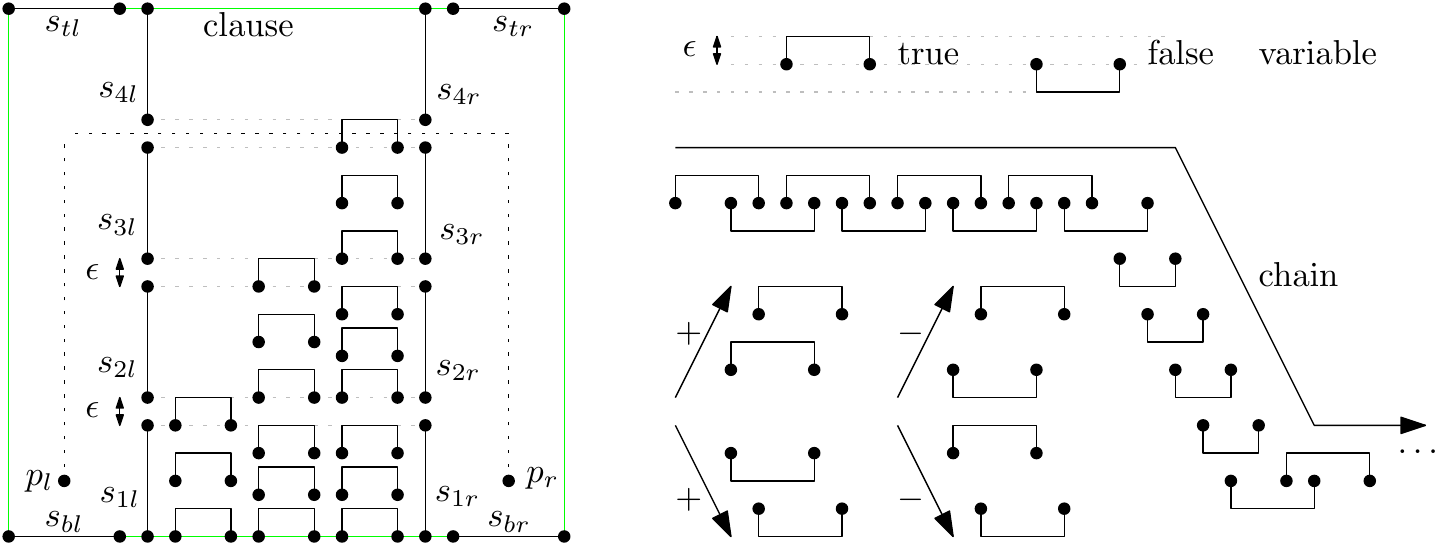}
 \caption{A clause, variable and chain gadget for reduction from planar 3-SAT.
Vertical propagation of true and false are unique, but $\sqcap$-buses are just uniquely propagated
in the top direction and $\sqcup$-buses are just uniquely propagated in the bottom direction.}
 \label{fig:gadgetsPlanar3SAT}
\end{figure}

A {\bf variable gadget} consists of two points $a_1,a_2$ of the same color on the same y-coordinate.
The value of the variable is true if the two points are connected with a $\sqcap$-bus and
the value of the variable is false if the two points are connected with a $\sqcup$-bus.
We use a variable gadget, referred to as a \df{chain link}, also as elements of chain gadgets.

A {\bf chain gadget} propagates the value of a chain link, which is actually a variable gadget,
to another chain link.
Let $a_1,a_2$ (respectively $b_1,b_2$) be the two points of the chain link at the beginning (respectively end)
of the chain. A chain gadget consists of $k$ chain links.

In a \emph{horizontal} chain gadget we place the points on a single
horizontal line in the order $a_1,b_1,a_2,b_2$ (respectively $b_1,a_1,b_2,a_2$)
for propagating to the right (respectively to the left).
If $a_1,a_2$ are connected with a $\sqcap$-bus,
then $b_1,b_2$ must be connected with a $\sqcup$-bus and the other way round.
This construction can be repeated until the chain consists of $k$ chain links.
The \emph{sign} of a horizontal chain is defined by $(-1)^{(1+k)}$.
Clearly if the sign is positive, then the first bus and the last bus are of the same type.
If the sign is negative, then the first bus and the last bus are different.

In a \emph{vertical} chain gadget we place $b_1$ below (respectively above) $a_1$ and $b_2$ below
(respectively above) $a_2$ on the same x-coordinate with a distance of $2\eps$
for propagating to the top (respectively to the bottom).
It is easy to check that in such a way we can only \emph{uniquely propagate} a $\sqcap$-bus to the top and
a $\sqcup$-bus to the bottom. It may happen that the type of buses change during a vertical propagation.
The \emph{sign} of a vertical chain is defined as $+1$. 

The sign of two chains, which are connected, will be multiplied.
If a literal in a clause appears positive, then the corresponding chain has sign $-1$, otherwise $+1$.

A {\bf clause gadget} consists of two main points $p_l,p_r$,
4 horizontal bounding segments $s_{tl}, s_{tr}, s_{bl}, s_{br}$,
8 vertical bounding segments $s_{1l}, s_{1r}, \dots, s_{4l}, s_{4r}$,
and 18 chain links, see a schematic illustration in Fig.~\ref{fig:gadgetsPlanar3SAT}.
We aim at satisfying the clause if and only if a bus connecting the main points can be drawn.

Within a bounding square $Q$ we place horizontal bounding segment $s_{tl}$ ($s_{tr}$)
in the top left (right) corner, and horizontal bounding segment $s_{bl}$ ($s_{br}$)
in the bottom left (right) corner. Above $s_{bl}$ ($s_{br}$) we place main point $p_l$ ($p_r$),
such that there is a normal to $s_{bl}$ ($s_{br}$) through $p_l$ ($p_r$) that is also
crossing $s_{tl}$ ($s_{tr}$). This construction prevents the bus connecting the main points
to be in the exterior of $Q$.

In $Q$ there are two vertical lines $l$ and $r$ that both separate $p_l$ from $p_r$.
We place the vertical bounding segments $s_{1x}, s_{2x}, s_{3x}, s_{4x}, x \in \{l,r\}$
in this order from bottom to top on line $x$ with $\eps$ distance between every pair
of consecutive segments. The resulting horizontal space between segment $s_{ix}$ and $s_{i+1x}$
is called $i$-th \emph{gap}, $i=1,2,3$. The gaps represent the literals in the clause.
This construction restricts the choices for the bus connecting
the main points to be precisely three.

Finally we place 9 chain links below the first gap, 6 chain links between the first and second gap and
3 chain links between the second and the third gap.
More specifically let $v_1,\dots,v_6$ be 6 vertical lines between $l$ and $r$ in this order from left to right.
We place 3 chain links on lines $v_1,v_2$ such that the first chain link has its points on the
boundary of $Q$, the last chain link has its points on the bottom boundary of the first gap
and the distance between every pair of chain link points is at most $2\eps$.
Similarly we place 6 chain links on lines $v_3,v_4$ such that the first chain link has its points on the
boundary of $Q$, the 4th chain link has its points on the top boundary of the first gap,
the last chain link has its points on the bottom boundary of the second gap,
and the distance between every pair of chain link points is at most $2\eps$.
In the same way we place 9 chain links on lines $v_5,v_6$ such that the first chain link has its points on the
boundary of $Q$, the 4th (7th) chain link has its points on the top boundary of the first (second) gap,
the last chain link has its points on the bottom boundary of the third gap,
and the distance between every pair of chain link points is at most $2\eps$.
We refer to the last chain link of lines $v_1,v_2$ (respectively lines $v_3,v_4$ and $v_5,v_6$)
as \emph{the chain link of the first gap} (respectively second and third gap).
This construction allows to block or open gaps from the bottom of $Q$.

Notice that it is easy to simulate vertical or horizontal segments with points as demonstrated in
Fig.~\ref{fig:example3SAT4}.

The construction of an instance of $(\sqcap,\sqcup)$-BEP$^{\eps}$ from an instance $I$ of planar 3-SAT
is done according to a planar drawing of the graph $G_I$.
We may assume that all variable vertices of $G_I$ are on a single horizontal line.
We use this line and place the variable gadgets according
to this order. The clause vertices above the variable vertices (top clauses)
are replaced by clause gadgets and the clause vertices below the variable vertices
(bottom clauses) are replaced by horizontally mirrored clause gadgets.
Finally we replace the edges of $G_I$ with chain gadgets.

The number of points needed to construct an instance of $(\sqcap,\sqcup)$-BEP$^{\eps}$
is polynomial in $n$ and $m$. Given a planar 3-SAT instance with $n$ variables and $m$ clauses,
the corresponding $(\sqcap,\sqcup)$-BEP$^{\eps}$ instance
has at most $\Oh(nm)$ points.
For a clause gadget we need precisely 118 points, for a variable
gadget we need 2 points and for chain gadgets we need 10 points plus the points needed
to surround other clause gadgets.
If an edge from a clause to variables vertically passes $k$ other clauses,
then we need $18k$ points for this construction. Since we have $\Oh(n+m)$ edges and $m$ clauses,
we might need $\Oh(nm+m^2)$ points for the edges.

\begin{figure}[tb]
 \centering
 \includegraphics[width=1.00\textwidth]{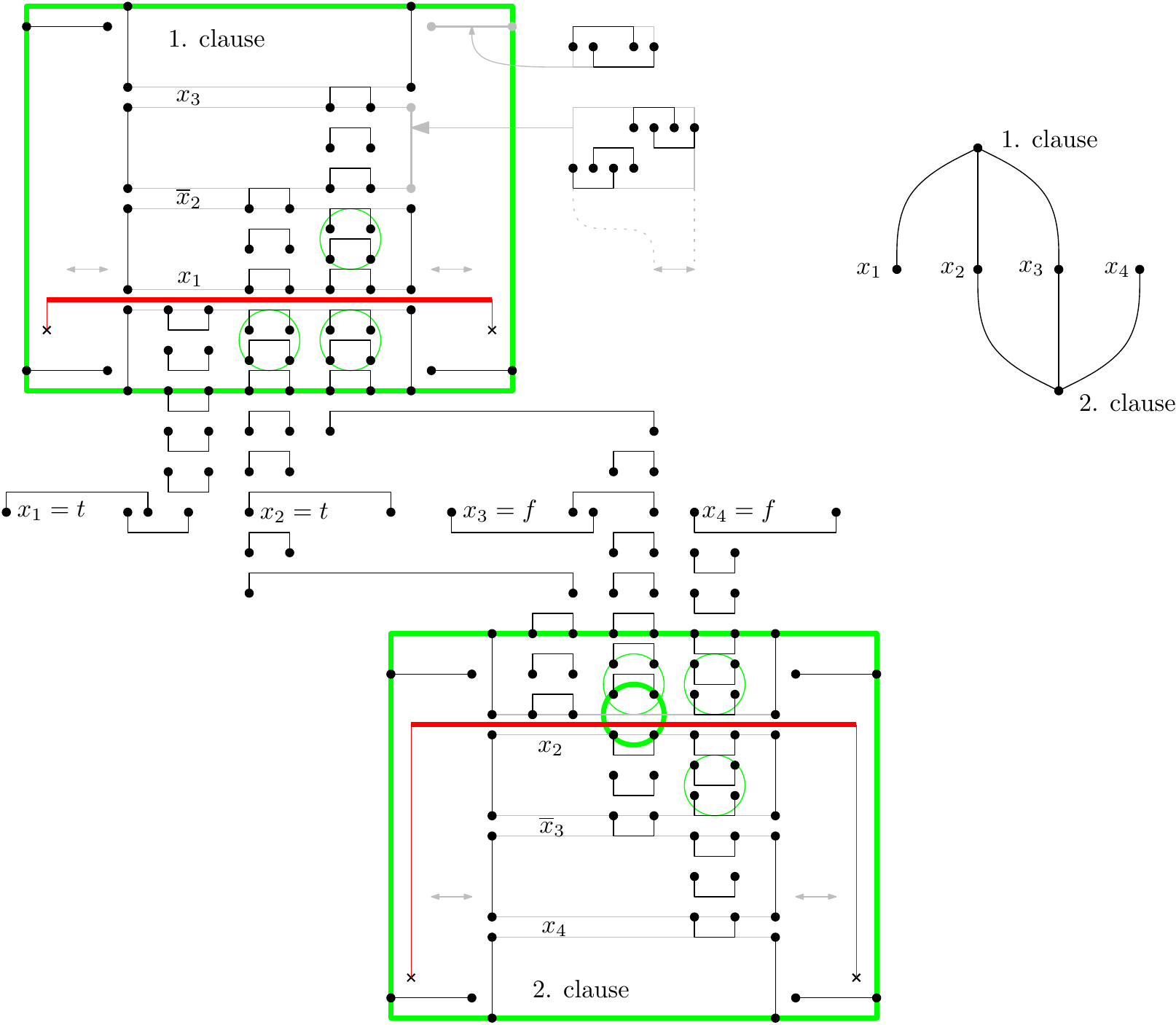}
 \caption{Point set instance constructed via gadgets for
$I = (x_1 \vee \overline{x}_2 \vee x_3) \wedge (x_2 \vee \overline{x}_3 \vee x_4)$. The red buses
indicate the truth assignment $x_3=x_4=f, x_1=x_2=t$. Clause gadgets are enclosed by a green rectangle.
The green thin circles indicate that
distances are less than $2\eps$, while the green thick circle indicates a change of bus types
during a not-unique vertical propagation of top buses to the bottom, that is, $x_3$ is meaningless for the
second clause.}
 \label{fig:example3SAT4}
\end{figure}

Fig.~\ref{fig:example3SAT4} shows how to use a variable several times: we stretch one chain link
for horizontal propagation and add a vertical chain gadget for vertical propagation.

\begin{theorem}
$(\sqcap,\sqcup)$-BEP$^{\eps}$ for 2 points per color is \NP-complete.
\end{theorem}

\begin{proof}
To show the membership of $(\sqcap,\sqcup)$-BEP$^{\eps}$ in the class \NP we observe that we have $n$ points and
between every pair of consecutive points we have a gap.
In every gap there can be possibly $n$ buses, that is, we have $(n-1)n$ slots,
where to place buses. So every slot represents a possibility to place a bus.
We can guess a drawing by choosing an order of the buses: all the drawings where
buses move within their gap are equivalent. To check if the order leads to a feasible
solution of $(\sqcap,\sqcup)$-BEP$^{\eps}$,
we apply the algorithm of Lemma~\ref{lem:correctness}.

We prove the hardness of $(\sqcap,\sqcup)$-BEP$^{\eps}$ by a reduction
from planar 3-SAT~\cite{DBLP:journals/siamcomp/Lichtenstein82}.
Let $I$ be an instance of the planar 3-SAT problem and let $P_I$ be the point set constructed
from the gadgets, that is, we replace in the planar graph representing $I$ every clause vertex by
a clause gadget, every variable vertex by a variable gadget and every edge by a chain gadget.
We prove next that $P_I$ admits a solution of $(\sqcap,\sqcup)$-BEP$^{\eps}$ $\Leftrightarrow$ $I$
has a satisfying truth assignment.

``$\Rightarrow$:''
If $P_I$ admits a solution, then in particular every pair of main points is connected.
Consider w.l.o.g. a top clause $c$ with literal $y$ corresponding to the gap
through which the main points are connected. We associate with $y$ the gap of $c$.
If the chain link of $y$ is a $\sqcup$-bus, then this bus is uniquely propagated
to the bottom and does not change its type.
If $y$ is a positive variable $x$, then by construction the chain has sign $-1$ and
the chain ends in a $\sqcap$-bus at the variable gadget, corresponding to $x$ being true.
If $y$ is a negated variable $\overline{x}$, then by construction the chain has sign $+1$ and
the chain ends in a $\sqcup$-bus at the variable gadget, corresponding to $x$ being false.
For bottom clauses the same argument holds horizontally mirrored.

``$\Leftarrow$:''
Assume we have a satisfying truth assignment for $I$. We explain how to construct a solution
for $(\sqcap,\sqcup)$-BEP$^{\eps}$. First for every variable being true we draw a $\sqcap$-bus,
while for every variable being false we draw a $\sqcup$-bus.
We propagate $\sqcap$-buses with $\sqcap$-buses to the top and to the bottom,
while we propagate $\sqcup$-buses with $\sqcup$-buses to the top and to the bottom.
A $\sqcap$-bus ($\sqcup$-bus) ends in a $\sqcap$-bus ($\sqcup$-bus)
if the variable in the top clause is negated (positive)
or the variable in the bottom clause is positive (negated).

We keep the type of buses in a vertical propagation as long as possible,
which can only be interrupted by a main bus. Then we change the type of buses
and the gap becomes blocked, although the variable is true and appears positive, or
the variable is false and appears negative in the clause.
Additionally this interrupting main bus indicates that this clause is already satisfied
and thus the variable of the interrupted chain
is irrelevant for the satisfiability of this particular clause.

By construction we get a feasible (planar) solution for the buses in $P_I$.

Finally we translate the construction into the ``no points share a coordinate" setting.
We may assume an underlying $k \times k$ grid with grid unit $\eps/2$ in the plane $\mathbb{R}^2$
so that all points have integer coordinates.
Let $p_1, \dots, p_n$ be the points ordered first by x-coordinate, then by y-coordinate.
We modify the x-coordinates by a shift $x(l) = x(l)+1$ for all $l \geq j$,
if $x(p_i) = x(p_j)$, as long as two points share the same x-coordinate.
We apply the same modification in y-direction with respect to the same order of points $p_1, \dots, p_n$.
Finally no points share a coordinate.

The properties of depending colors stay the same, since the topological operation is just a stretch.
Clearly chain links are dependent before the stretch, if and only if they are
dependent after the stretch.
\end{proof}

We consider as an example the instance
$I = (x_1 \vee \overline{x}_2 \vee x_3) \wedge (x_2 \vee \overline{x}_3 \vee x_4)$
of planar 3-SAT. Clearly the clause-variable graph is planar.
Fig.~\ref{fig:example3SAT4} illustrates the point set created from the instance $I$.

We can adopt the same construction when additionally using center buses.
Now some vertical segments can be modeled by using just two points.
In a clause gadget, we move one of the main points from bottom to top such that the bus connecting
the main points is necessarily a center bus. The remaining parts are the same.
Also for center buses we need $\eps$ as input for the minimum distance of buses to their points.
Notice that a bus $c$ and a point $p$ of different color $c \neq c(p)$ may be closer than $\eps$,
as well as two buses $c,c'$ may be closer than $\eps$.

\begin{theorem}
BEP$^{\eps}$ for 2 points per color is \NP-complete.
\end{theorem}

\section{Conclusion and Future Work}

We studied bus embeddability, where a set of colored points is covered by
a set of horizontal buses, one per color and without crossings.
We described an ILP and an FPT algorithm for the general problem and presented polynomial-time
algorithms for several restricted versions.
The general problem is shown to be \NP-complete even for two points per color
when points may not lie on buses.

It is still open to determine the complexity of BEP in the following cases:
\begin{itemize}
 \item BEP using only center-buses;
 \item $(\sqcap, \sqcup)$-BEP, that is, BEP without center-buses;
 \item diagonal BEP with more than 2 points per color;
 \item general BEP (in our construction, we use an extra $\eps$ as a parameter).
\end{itemize}

A natural generalization would be to allow both horizontal and vertical buses,
as in~\cite{DBLP:journals/corr/abs-cs-0609127,DBLP:conf/ciac/BruckdorferFK13}.
Another variant might be to consider multi-colored points, where a
point has to be connected either to all the buses of its corresponding colors,
or to at least one of them.
For point sets that have no solution for BEP with only one bus per color,
we may allow more than one bus or bound the number of crossings.
Possible objectives in these scenarios are
to minimize the total number of buses over all colors,
to minimize the total number of buses,
or to minimize the total number of buses if each tree can connect $\leq$ k unicolored points.
These objectives are even interesting if a solution to BEP exists.

\bibliographystyle{splncs03}
\bibliography{abbrv,references}

\end{document}